\documentclass{amsart}

\title[Singular KP solitons and the Go-diagrams]
{On singular solitons of the KP equation and the Go-diagrams}
\author{Shilong Huang}
\date{\today}
\address{College of Mathematics and Systems Science, Shandong University of Science and Technology, Qingdao, 266590, China}
\email{shilonghhqu@163.com}
\subjclass[2000]{}

\usepackage{url}
\usepackage{color}
\usepackage{rotating}
\newcommand{\rotxc}[1]{\begin{sideways}#1\end{sideways}}
\newcommand{\invert}[1]{\rotxc{\rotxc{#1}}}

\def\Le{\hbox{\invert{$\Gamma$}}}
\countdef\x=23
\countdef\y=24
\countdef\z=25
\countdef\t=26

\def\tbox(#1,#2)#3{
\x=#1 \y=#2
\multiply\x by 12
\multiply\y by 12
\z=\x \t=\y
\advance\z by 12
\advance\t by 12
\psline(\x,\y)(\x,\t)(\z,\t)(\z,\y)(\x,\y)
\advance\x by 6
\advance\y by 6
\rput(\x,\y){{\bf #3}}}

  \usepackage{paralist}
  \usepackage{graphics} 
  \usepackage{epsfig} 
\usepackage{graphicx}
\usepackage{epstopdf}
\usepackage{subfigure}
\usepackage{array}

\usepackage{float}
\usepackage{caption}
\usepackage{CJK}

\usepackage{amssymb}
\usepackage{graphicx}
\usepackage{epstopdf}
\usepackage{amscd}
\usepackage{appendix}
\usepackage{graphics}
\usepackage{tikz}

\usepackage{mathdots}
\def\proof{\par{\it Proof}. \ignorespaces}
\def\endproof{{\ \vbox{\hrule\hbox{%
     \vrule height1.3ex\hskip0.8ex\vrule}\hrule }}\par}

\theoremstyle{definition}

\theoremstyle{remark}

\numberwithin{equation}{section}

\let\trueint=\int
\let\truesum=\sum
\def\int{\mathop{\textstyle\trueint}\limits}
\def\sum{\mathop{\textstyle\truesum}\limits}
\let\<=\langle
\let\>=\rangle

\def\Gr{\text{Gr}}

\def\v{\mathbf{v}}
\def\w{\mathbf{w}}
\def\t{\mathbf{t}}
\def\0{\mathbf{0}}

\def\edge{\ar@{-}}
\def\dedge{\ar@{.}}

\usepackage{amsmath, amsthm, amssymb, amsbsy,amsfonts,amscd}
\usepackage[mathscr]{eucal}
\usepackage{epsfig}
\usepackage{young}
\newtheorem{theorem}{Theorem}[section]

\newtheorem{definition}[theorem]{Definition}
\newtheorem{proposition}[theorem]{Proposition}
\newtheorem{lemma}[theorem]{Lemma}

\newtheorem{example}[theorem]{Example}
\newtheorem{corollary}[theorem]{Corollary}

\newtheorem{remark}[theorem]{Remark}

\usepackage{amsfonts}
\usepackage{xy}
\usepackage{amssymb}
\usepackage{amsmath}

\newcommand{\R}{\mathbb R}

\DeclareMathOperator{\In}{in}
\DeclareMathOperator{\Out}{out}

\definecolor{revisionyellow}{RGB}{210,150,0}

\newcommand{\thmrefer}[1]{\renewcommand\thetheorem
  {\protect\ref{#1}}\addtocounter{theorem}{-1}}

\xyoption{all}
\CompileMatrices

\begin{document}

\begin{abstract}
It has been proven that real and regular soliton solutions of the KP equation are classified in terms of the totally nonnegative Grassmannian. It is well known that vertex operators can be used to construct soliton solutions. In this paper, we consider several regular soliton solutions and study their combinations through products of vertex operators. In general, the resulting solutions become \emph{singular}. Totally nonnegative elements are parametrized by the $\Le$-diagrams introduced by Postnikov. We show that the resulting singular solutions can be parametrized by Go-diagrams, which extend $\Le$-diagrams and arise in the Deodhar decomposition of the Grassmannian.
\end{abstract}

\maketitle

\tableofcontents

\section{Introduction}
The Kadomtsev-Petviashvili (KP) equation
\begin{equation}\label{eq:KP}
(-4u_{t}+6uu_{x}+u_{xxx})_{x}+3u_{yy}=0
\end{equation}
admits soliton solutions expressed in terms of a $\tau$-function
\begin{equation}\label{eq:u-tau}
u(x,y,t)=2\left(\ln \tau(x,y,t)\right)_{xx}.
\end{equation}
Each KP soliton solution is determined by a pair of data $(\kappa,A)$, where $\kappa=(\kappa_1<\kappa_2<\cdots<\kappa_M)$ is an ordered set of spectral parameters and $A$ is an $N\times M$ full rank matrix. The associated $\tau$-function can be expressed as
\begin{equation*}
\tau(x,y,t)=\sum_{I}\Delta_I(A)E_I(x,y,t),
\end{equation*}
where $I$ runs over all $N$-element subsets of $[M]:=\{1,\ldots,M\}$, $\Delta_I(A)$ denotes the Pl\"ucker coordinate indexed by $I$, namely the corresponding maximal minor of $A$, and
\begin{equation*}
E_I(x,y,t)=\prod_{k>l}(\kappa_{i_k}-\kappa_{i_l})E_{i_1}\cdots E_{i_N},\qquad E_{j}(x,y,t)=e^{\xi_{j}(x,y,t)}=e^{\kappa_j x+\kappa_j^2 y+\kappa_j^3 t}.
\end{equation*}
This gives a parametrization of KP soliton solutions by the Grassmannian $\Gr(N,M)$, i.e. the space of $N$-dimensional subspaces of $\mathbb{R}^M$ \cite{CK:08,CK:09}.

It is known in \cite{KW:13} that a KP soliton solution is real and regular \emph{if and only if} $A\in\Gr(N,M)_{\geq0}$, where $\Gr(N,M)_{\ge0}$ denotes the totally nonnegative Grassmannian, i.e. all Pl\"ucker coordinates of $A$ are nonnegative. As shown in \cite{CK:08, KW:14}, the asymptotic structure of such solutions is encoded by a permutation $\pi(A)\in S_M$, which determines the pairing of line-solitons. These line-solitons are labeled by pairs $(i,j)$ with $i<j$, and are called line-solitons of type $[i,j]$.

Regular KP soliton solutions admit a combinatorial description via \emph{$\Le$-diagrams}, namely decorated Young diagrams with certain boxes filled by circles, which parametrize the positroid cells in Postnikov's stratification of the totally nonnegative Grassmannian \cite{P:06}. The rows and columns of a $\Le$-diagram are labeled by the pivot and nonpivot indices of $A$, respectively, so that each box is indexed by a pair $(i,j)$ with $i<j$. Moreover, these diagrams satisfy the \emph{$\Le$-property}: if a box contains a circle, then all boxes weakly to its left or above also contain circles \cite{KW:13}.

However, singular KP soliton solutions are much less understood. For the KdV equation, i.e. the $2$-reduction of the KP hierarchy, combinations of regular one-soliton solutions are always regular. Surprisingly, for the Boussinesq equation corresponding to the $3$-reduction, combining two regular one-soliton solutions can produce a singular solution \cite{HY:24}. This is the first case where total nonnegativity is lost under combinations of regular soliton solutions. As shown in \cite{HY:24}, \emph{vertex operators} can be used to construct regular soliton solutions, and their combinations are realized by products of vertex operators. In that work, the regularity of such combinations was characterized by a non-crossing condition.

In this paper, we study the singular case arising when this condition fails. More precisely, we study combinations of regular KP soliton solutions and give a combinatorial description of the resulting singular solutions. We begin with the simplest case of one-soliton solutions. Let $V_{ij}=e^{a_{ij}X_{ij}}$ be the vertex operator associated with the regular line-soliton $[i,j]$ with $a_{ij}>0$. The corresponding $\tau$-function is given by
\begin{align*}
\tau=V_{ij}\cdot 1=1+e^{\xi_{j}-\xi_{i}+\ln a_{ij}},
\end{align*}
where $X_{ij}=X(\kappa_i,\kappa_j)$ denotes the KP vertex operator defined in Section~\ref{sec:3}.
The corresponding soliton is represented in the Grassmannian by $A=(1,a_{ij})\in\Gr(1,2)_{\geq0}$, where \(i\) is the pivot index and \(j\) is the nonpivot index. We now consider two regular line-solitons of types $[i,j]$ and $[k,l]$, where $i,j,k,l$ are pairwise distinct. Their combination gives the two-soliton $\tau$-function
\begin{equation*}
\tilde{\tau}=V_{ij}\cdot V_{kl}\cdot 1=1+e^{\xi_{j}-\xi_{i}+\ln a_{ij}}+e^{\xi_{l}-\xi_{k}+\ln a_{kl}}+\frac{(\kappa_{i}-\kappa_{k})(\kappa_{j}-\kappa_{l})}{(\kappa_{i}-\kappa_{l})(\kappa_{j}-\kappa_{k})}e^{(\xi_{j}-\xi_{i})+(\xi_{l}-\xi_{k})+\ln a_{ij}+\ln a_{kl}}.
\end{equation*}
The coefficient of the last term is the cross-ratio factor arising from the product of the two vertex operators (see Section~\ref{sec:3} for details). Its sign depends on the relative ordering of $i,j,k,l$. Without loss of generality, assume that $i<k$. Then there are three possible relative positions of the pairs $(i,j)$ and $(k,l)$ shown below.
\begin{figure}[H]
  \begin{minipage}[t]{1\linewidth}
  \centering
  \includegraphics[height=1.5cm,width=11cm]{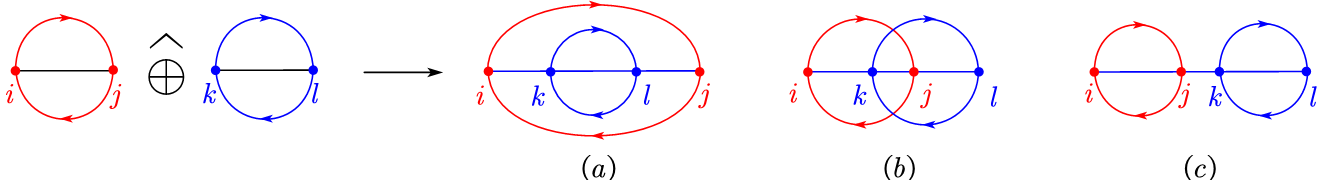}
  \end{minipage}
\end{figure}
\noindent
The associated permutation is $\pi=(i,j)(k,l)$, and the cross-ratio factor is negative only in the crossing case (b), namely $i<k<j<l$. This implies that case (b) produces a singular solution, whereas cases (a) and (c) give regular solutions.

We now illustrate these three cases diagrammatically. Since a one-soliton corresponds to a single empty box, the relative positions of the boxes associated with $(i,j)$ and $(k,l)$ force the appearance of additional boxes, which are filled with circles as follows.
\begin{figure}[H]
  \begin{minipage}[t]{1\linewidth}
  \centering
  \includegraphics[height=1.6cm,width=10.8cm]{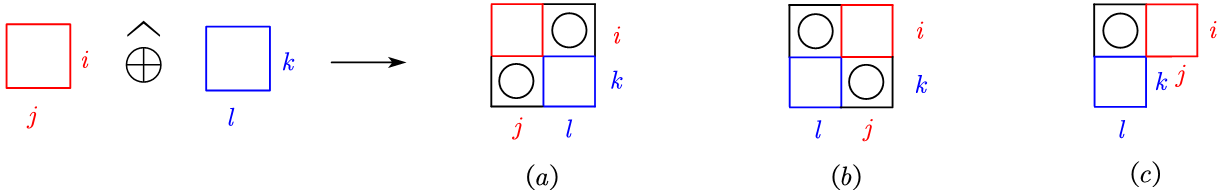}
  \end{minipage}
\end{figure}
\noindent In cases (a) and (c), the $\Le$-property is satisfied, and hence these diagrams are $\Le$-diagrams. Accordingly, all circles are represented by white stones. In contrast, the $\Le$-property fails in case (b), so the crossing case is not captured by $\Le$-diagrams. This failure raises a natural question: how can one describe combinatorially the crossing case (b), which corresponds to a singular KP soliton solution? As we will show in Section~\ref{sec:4}, this case is characterized by \emph{Go-diagrams} \cite{KW:13}, which extend $\Le$-diagrams. In Go-diagrams, circles are represented by either white or black stones, and determining their types is one of the main goals of this paper.

The above construction extends naturally to combinations of more general regular KP soliton solutions. Each regular soliton corresponds to a $\Le$-diagram. After embedding the corresponding $\Le$-diagrams into the Young diagram determined by the permutation of the combined solution, their relative positions force the appearance of additional boxes. The resulting diagram is called a composite diagram. We show that these additional boxes admit canonical fillings compatible with the \emph{Deodhar decomposition} of the Grassmannian, which is a refinement of the Bruhat decomposition \cite{Deodhar,MR:04}. Our main theorem (Theorem~\ref{th:main}) characterizes these canonical fillings.
\begin{theorem}
Let
\begin{align*}
\tilde{A}=\widehat{\bigoplus}_{k=1}^KA_{k}\in\text{Gr}(N,M)
\end{align*}
be the $\kappa$-direct sum of totally nonnegative matrices $A_{k}\in\text{Gr}(n_{k},m_{k})_{\geq0}$ with $N=n_{1}+\ldots+n_{K}$ and $M=m_{1}+\ldots+m_{K}$. The associated permutation satisfies
\begin{align*}
\pi(\tilde{A})=\prod_{k=1}^{K}\pi(A_{k}).
\end{align*}
Consider the composite diagram obtained by embedding the $\Le$-diagram associated with each $A_k$ into the Young diagram determined by $\pi(\tilde{A})$. Then this composite diagram canonically determines a Go-diagram, characterized as follows:
\begin{itemize}
\item the embedded $\Le$-diagrams associated with the $A_k$ remain unchanged;
\item boxes whose south-east corner lies on the boundary of the composite diagram, together with the consecutive boxes to their left or above, are filled with white stones;
\item the remaining boxes are determined by pairs of line-solitons $[i,j]$ and $[k,l]$ with $i<k$. Each such box is filled with a stone: it is black if $i<k<j<l$ and white otherwise.
\end{itemize}
\end{theorem}
This theorem shows that combinations of regular soliton solutions can be parametrized by Go-diagrams.

The paper is organized as follows. In Section \ref{sec:2}, we review the Grassmannian description of KP soliton solutions, emphasizing the role of Pl\"ucker signs in determining regularity. In Section \ref{sec:3}, we introduce the vertex operator construction and the associated $\kappa$-direct sum. We explain how cross-ratio factors govern the signs of the Pl\"ucker coordinates of the $\kappa$-direct sum, leading to singular solutions. In Section \ref{sec:4}, we construct and analyze the composite diagram associated with the $\kappa$-direct sum via the Deodhar decomposition, and establish its canonical Go-diagram filling. Finally in Section \ref{sec:5}, we present examples of the Go-diagram construction, including Hirota $N$-soliton solutions with associated permutation of the form $\pi=\prod_{k=1}^{N}(i_k,j_k)$.

\medskip

\section{Grassmannian description of KP soliton solutions}\label{sec:2}
In this section, we recall the Grassmannian description of real KP soliton solutions, focusing on the expansion of the $\tau$-function, the role of Pl\"ucker signs in determining regularity, and the permutation encoding the asymptotic pairing of line-solitons.

Let $\kappa=(\kappa_1<\ldots<\kappa_M)\in\R^M$ and $A$ be an $N\times M$ real matrix with full rank. We define
\begin{align*}
E_j(x,y,t)=e^{\xi_j(x,y,t)}\quad\quad\text{with}\quad\quad \xi_{j}=\kappa_{j}x+\kappa_{j}^{2}y+\kappa_{j}^{3}t.
\end{align*}
For each $N$-element subset $I=\{i_1<i_2<\ldots<i_N\}\subset[M]:=\{1,2,\ldots,M\}$, we set
\begin{equation}\label{eq:EI}
E_{I}(x,y,t)=\prod_{k<l}(\kappa_{i_l}-\kappa_{i_k})
\exp\left(\xi_{i_1}+\cdots+\xi_{i_N}\right),
\end{equation}
where the Vandermonde factor is positive since $\kappa_1<\ldots<\kappa_M$. Thus each term $E_I(x,y,t)>0$ for all $(x,y,t)$, and the $\tau$-function associated with $(\kappa,A)$ admits the expansion
\begin{equation}\label{eq:tauPl}
\tau(x,y,t)=\sum_{I\in\binom{[M]}{N}}\Delta_{I}(A)E_{I}(x,y,t),
\end{equation}
where $\Delta_I(A)$ denotes the Pl\"ucker coordinates of $A$ indexed by $I$. The signs of the terms in the $\tau$-function are entirely determined by the Pl\"ucker coordinates $\Delta_I(A)$, which in turn control the regularity of the solution.

\begin{example} A one-soliton solution of the KP equation \eqref{eq:KP} is determined by a pair $\kappa=(\kappa_i,\kappa_j)$ and $A=(1,a)\in\text{Gr}(1,2)$.
\begin{itemize}
\item[(a)] If $a>0$, the $\tau$-function is given by
\begin{align*}
\tau(x,y,t)=E_i(x,y,t)+aE_j(x,y,t)=2e^{\frac{\xi_{i}+\xi_{j}+\ln a}{2}}\cosh(\frac{\xi_i-\xi_j-\ln a}{2}).
\end{align*}
This gives a regular line-soliton of type $[i,j]$.
\item[(b)] If $a<0$, the $\tau$-function becomes
\begin{align*}
\tau(x,y,t)=E_i(x,y,t)+aE_j(x,y,t)=2e^{\frac{\xi_{i}+\xi_{j}+\ln |a|}{2}}\sinh(\frac{\xi_i-\xi_j-\ln |a|}{2}).
\end{align*}
Since the $\tau$-function vanishes along a line, the corresponding solution is singular.
\end{itemize}
More generally, a line-soliton $[i,j]$ arises from the balance between two adjacent dominant exponential terms indexed by $I$ and $J$ with $I\setminus J=\{i\}$ and $J\setminus I=\{j\}$. It is regular if $\Delta_I(A)\Delta_J(A)>0$, and singular otherwise.
\end{example}
It is known that real and regular soliton solutions arise if and only if the matrix $A\in\text{Gr}(N,M)$ belongs to the \emph{totally nonnegative} Grassmannian \cite{KW:13}, defined by
\begin{align*}
\Gr(N,M)_{\ge0}:=\left\{A\in \Gr(N,M):\Delta_I(A)\geq0~\text{for all}~I\in\binom{[M]}{N}\right\}.
\end{align*}
Every $A\in\text{Gr}(N,M)$ admits a unique representative in reduced row echelon form (RREF). Throughout this paper we assume that $A$ is irreducible in RREF, meaning that each row contains a nonzero entry besides the pivot and there are no zero columns. The pivot set determines the lexicographically minimal nonzero Pl\"ucker coordinate \cite{K:17}.

For a regular soliton solution of the KP equation, the dominant exponentials of the $\tau$-function determine the asymptotic line-solitons. The corresponding asymptotic pairing is encoded by a permutation $\pi(A)$, characterized by the following theorem \cite{CK:08, CK:09}.
\begin{theorem}\label{th:pairing}
Let $\{i_1,\ldots,i_N\}$ be the pivot set and $\{j_1,\ldots,j_{M-N}\}$ the nonpivot set of $A\in\text{Gr}(N,M)_{\geq0}$.
There exists a unique permutation $\pi\in S_M$ such that
\begin{itemize}
\item[(a)] For $y\gg 0$, there are $[i_n,\pi(i_n)]$-solitons with $\pi(i_n)>i_n$ for $n=1,\ldots,N$.
\item[(b)] For $y\ll 0$, there are $[\pi(j_m),j_m]$-solitons with $\pi(j_m)<j_m$ for $m=1,\ldots,M-N$.
\end{itemize}
\end{theorem}
In particular, each pair $(i,\pi(i))$ is represented by a chord connecting the boundary labels $i$ and $\pi(i)$, giving a chord diagram. The chord joining $i$ and $\pi(i)$ is drawn above the line if $\pi(i)>i$, and below the line if $\pi(i)<i$.
\begin{example}
Let $\pi=(1,5,3)(2,7,8)(4,6)\in S_8$ be a permutation. The associated chord diagram is shown below.
\begin{figure}[H]
  \begin{minipage}[t]{1\linewidth}
  \centering
  \includegraphics[height=2cm,width=6cm]{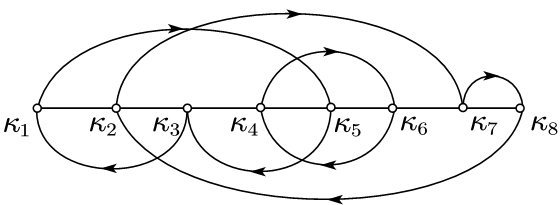}
  \end{minipage}
\end{figure}
\end{example}
\begin{lemma}\cite{C:07}\label{lem:chord}
The number of free parameters in $A\in\text{Gr}(N,M)_{\geq0}$ is read off from the chord diagram and is given by
\begin{align*}
\{\text{\# of pivots}\}+\{\text{\# of crossings}\}+\{\text{\# of cusps in the lower part}\},
\end{align*}
which coincides with the dimension of the positroid cell associated with $\pi(A)$.
\end{lemma}
In the next section, we use vertex operators to construct regular KP soliton solutions and study combinations of such regular solutions through products of vertex operators.

\section{Vertex operator construction and $\kappa$-direct sum}\label{sec:3}
We review the vertex operator construction \cite{DKJM:83, Na:23}, which provides a systematic way to generate and combine soliton solutions.

Let $\kappa=(\kappa_{1}<\ldots<\kappa_{M})$ be an ordered set of real parameters. The vertex operators are defined by
\begin{equation*}
X(\kappa_i, \kappa_j)=\exp(\xi_{j}-\xi_{i})\exp\left(-\left(\tilde{\partial}(\kappa_{j})-\tilde{\partial}(\kappa_{i})\right)\right),
\end{equation*}
where
\begin{align*}
\xi_{k}=\sum_{n=1}^{\infty}\kappa_{k}^n t_n\quad\quad\text{and}\quad\quad
\tilde{\partial}(\kappa_{k})=\sum_{n=1}^{\infty}\frac{\kappa_{k}^{-n}}{n}\frac{\partial}{\partial t_n}.
\end{align*}

\begin{lemma}\label{lem:X}
The vertex operators satisfy
\begin{equation*}
X(\kappa_{i},\kappa_{j})X(\kappa_{k},\kappa_{l})=C_{[(i,j),(k,l)]}: X(\kappa_{i},\kappa_{j})X(\kappa_{k},\kappa_{l}):,
\end{equation*}
where $C_{[(i,j),(k,l)]}$ is the cross-ratio factor given by
\begin{equation*}
C_{[(i,j),(k,l)]}=\frac{(\kappa_{i}-\kappa_{k})(\kappa_{j}-\kappa_{l})}{(\kappa_{i}-\kappa_{l})(\kappa_{j}-\kappa_{k})},
\end{equation*}
whose sign depends only on the relative ordering of the spectral parameters. Here $:\cdot:$ denotes normal ordering. In particular,
\begin{align*}
&X(\kappa_{i},\kappa_{j})X(\kappa_{k},\kappa_{l})=X(\kappa_{k},\kappa_{l})X(\kappa_{i},\kappa_{j}), \qquad\text{if}\quad (\kappa_i-\kappa_l)(\kappa_j-\kappa_k)\ne 0,\\
&X(\kappa_{i},\kappa_{j})X(\kappa_{k},\kappa_{l})=0,\qquad \text{if $\kappa_{i}=\kappa_{k}$ or $\kappa_{j}=\kappa_{l}$}.
\end{align*}
\end{lemma}
\subsection{The vertex operators for totally nonnegative matrices}
Let $A\in\text{Gr}(n,m)_{\geq0}$ with $1\leq n\leq m-1$ be a totally nonnegative matrix. We denote the pivot set by $\mathcal{I}_{P}[A]=\{i_{1}<\ldots<i_{n}\}$. For each row $k$, let $a_{k,j}$ denote the nonzero non-pivot entries of $A$. Following \cite{HY:24}, the vertex operator associated with $A$ is defined as the product of $X(\kappa_{i_k},\kappa_j)$ determined by these nonzero entries:
\begin{equation}\label{eq:VO}
V[A]:=\exp\left(\sum_{k=1}^n\sum_{j}b_{k,j}X(\kappa_{i_k},\kappa_{j})\right)=\prod_{k=1}^{n}\prod_{j}\left(1+b_{k,j}X(\kappa_{i_k},\kappa_{j})\right),
\end{equation}
where both the sum and the product run over all indices $j$ for which $a_{k,j}\neq 0$, and the coefficient $b_{k,j}$ is given by
\begin{align*}
b_{k,j}=a_{k,j}\prod_{p\ne k}\frac{\kappa_{i_p}-\kappa_{j}}{\kappa_{i_p}-\kappa_{i_k}}.
\end{align*}
Note that in $X(\kappa_{i_k},\kappa_{j})$, the index $i_{k}$ lies in $\mathcal{I}_{P}[A]$ while $j\notin\mathcal{I}_{P}[A]$, so the spectral parameters are distinct. Hence these operators commute by Lemma \ref{lem:X}.
\begin{remark}\cite{HY:24}
For a totally nonnegative matrix $A\in\text{Gr}(n,m)_{\geq0}$, all coefficients $b_{k,j}$ are strictly positive. Thus the signs of the nonzero entries $a_{k,j}$ are determined by the positivity condition $b_{k,j}>0$.
\end{remark}

With this choice of \(b_{k,j}\), the normalized form of the $\tau$-function \eqref{eq:tauPl} can be written as
\begin{align*}
\frac{\tau(x,y,t)}{E_{\mathcal{I}_{P}[A]}(x,y,t)}=V[A]\cdot 1,
\end{align*}
where $E_{\mathcal{I}_{P}[A]}$ is the exponential term corresponding to the pivot set. Since multiplication by a positive factor does not affect the solution $u(x,y,t)=2(\ln\tau(x,y,t))_{xx}$, we may equivalently write
\begin{equation}\label{eq:nortau}
\tau_{A}(x,y,t)=V[A]\cdot 1.
\end{equation}
(See \cite{HY:24} for details). Note that the $\tau$-function $\tau_A(x,y,t)$ is of Grammian form and can be identified with the $M$-theta function introduced in \cite{K:24}. In its expansion, terms corresponding to the same exponential phase combine together, reflecting the Pl\"ucker structure of $A$.

\subsection{Products of vertex operators $V[A_{k}]$ and $\kappa$-direct sums}
We now consider products of the vertex operators \(V[A_k]\) associated with different totally nonnegative matrices $A_{k}$. As discussed in \cite{HY:24}, cross-ratio factors arising between the operators $X(\kappa_{i_k},\kappa_{j})$ appearing in different \(V[A_k]\) may lead to singular soliton solutions.

Let $A_k\in\text{Gr}(n_k,m_k)_{\ge0}$, $k=1,\ldots,K$, be totally nonnegative matrices associated with vertex operators as in \eqref{eq:VO}, which we regard as building blocks of a composite solution. We set $N=n_{1}+\ldots+n_{K}$ and $M=m_{1}+\ldots+m_{K}$. The $\tau$-function generated by their product is defined by
\begin{equation}\label{eq:tau-VO}
\tau_{A_{1},\ldots,A_{K}}(x,y,t)=V[A_1]\cdot V[A_2]\cdots V[A_K]\cdot 1.
\end{equation}

For each block $A_k$, let $\mathcal{I}[A_k]$ be a part of a \emph{partition} of $[M]$ with $|\mathcal{I}[A_k]| = m_k$, and assign spectral parameters $\{\kappa_{i}:~i\in\mathcal{I}[A_{k}]\}$ to $A_{k}$. Reordering all parameters in increasing order gives
\begin{align*}
\kappa:=\text{ord}\left(\bigsqcup_{k=1}^K\{\kappa_{i}:~i\in\mathcal{I}[A_{k}]\}\right)=(\kappa_1<\kappa_2<\cdots<\kappa_M).
\end{align*}
This induces an \emph{order-preserving bijection}
\begin{align*}
\phi:~\bigsqcup_{k=1}^{K}\mathcal{I}[A_{k}]\rightarrow [M],
\end{align*}
so that the reordered index set $\hat{\mathcal{I}}[A_{k}]:=\phi(\mathcal{I}[A_{k}])$ respects the increasing order of spectral parameters $\kappa_1<\cdots<\kappa_M$.

We then form an $N\times M$ matrix by assembling the blocks $A_{k}$ with $k=1,\ldots,K$ according to the reordered index sets $\hat{\mathcal{I}}[A_{k}]$. This construction is called the \emph{$\kappa$-direct sum}.
\begin{definition}\label{def:SumA}
The $\kappa$-direct sum of the matrices $A_{k}$ is the block matrix
\[
\widehat{\bigoplus}_{k=1}^KA_{k}\in\text{Gr}(N,M)
\]
obtained by placing each block $A_{k}$ into the columns indexed by $\hat{\mathcal{I}}[A_{k}]$ and arranging the rows so that the resulting matrix is in RREF. Notice that
in general it is not totally nonnegative. This matrix is associated with the $\tau$-function defined in \eqref{eq:tau-VO}.
\end{definition}

Each totally nonnegative block $A_k$ has an associated permutation $\pi(A_k)$ describing its asymptotic soliton structure (see Theorem~\ref{th:pairing}). Under the bijection $\phi$, we regard $\pi(A_k)$ as a permutation on $\hat{\mathcal{I}}[A_k]$. In cycle notation, we write
\[
\pi(A_k)=\prod_{p=1}^{P_k}(j^{(p)}_1,\ldots,j^{(p)}_{k_p}),
\]
where $P_k$ is the number of cycles, and each $(j^{(p)}_1,\ldots,j^{(p)}_{k_p})$ is a $k_p$-cycle with entries in $\hat{\mathcal{I}}[A_{k}]$. As shown in \cite[Corollary 3.6]{HY:24}, the $\kappa$-direct sum preserves the dominant exponential structure associated with each block, regardless of whether it is totally nonnegative. Thus the permutation associated with the $\kappa$-direct sum is given by
\begin{equation}\label{eq:kapp}
\pi(\widehat{\bigoplus}_{k=1}^KA_k)=\prod_{k=1}^{K}\pi(A_{k}).
\end{equation}
Since the index sets are disjoint, the dimension contributions from different blocks are independent, and hence the positroid cell dimensions are additive. However, the dimension of each totally nonnegative block can be read off from its associated chord diagram (see Lemma \ref{lem:chord}). Each crossing between chords from different blocks introduces an additional dimension, so the corresponding $\kappa$-direct sum cannot remain totally nonnegative.

\begin{remark}\label{re:c}
Although each block \(A_k\) is totally nonnegative, the cross-ratio factors arising between the operators $X(\kappa_{i_k},\kappa_{j})$ appearing in different \(V[A_k]\) are not controlled by the total nonnegativity of the individual blocks. As a result, the Pl\"ucker coordinates of the $\kappa$-direct sum may change sign, leading to singular soliton solutions.
\end{remark}

\begin{example}\label{ex:l5}
We illustrate the $\kappa$-direct sum construction and compare the roles of cross-ratio factors within a single vertex operator and between different vertex operators. Consider two vertex operators $V[A_{1}]$ and $V[A_{2}]$ with $|\mathcal{I}[A_{1}]|=4$ and $|\mathcal{I}[A_{2}]|=3$, whose reordered index sets are
\[
\hat{\mathcal{I}}[A_1]=\{1,3,4,6\},\qquad \hat{\mathcal{I}}[A_2]=\{2,5,7\}.
\]
As a concrete example, consider the permutations in cycle notation,
\begin{align*}
\pi(A_1)=(1,3,6,4),\qquad \pi(A_2)=(2,7,5).
\end{align*}
The associated matrices $A_{1}\in\text{Gr}(2,4)_{\geq0}$ and $A_{2}\in\text{Gr}(1,3)_{\geq0}$ are given by
\begin{align*}
A_{1}=\left(
\begin{array}{cccc}
1 & 0 & a_{1,4} & a_{1,6}\\
0 & 1 &  a_{3,4} & a_{3,6}\\
\end{array}
\right),\qquad
A_{2}=\begin{pmatrix}
1 & a_{2,5} &a_{2,7}
\end{pmatrix},
\end{align*}
where we impose the relation $a_{1,4}a_{3,6}-a_{3,4}a_{1,6}=0$. From these nonzero entries we construct the associated vertex operators
\begin{align*}
V[A_1]&=\exp\left(b_{1,4}X(\kappa_1,\kappa_4)+b_{1,6}X(\kappa_1,\kappa_6)+
b_{3,4}X(\kappa_3,\kappa_4)+b_{3,6}X(\kappa_3,\kappa_6)\right),\\
V[A_2]&=\exp\left(b_{2,5}X(\kappa_2,\kappa_5)+b_{2,7}X(\kappa_2,\kappa_7)\right),
\end{align*}
where
\begin{align*}
b_{1,4}=a_{1,4}\frac{\kappa_3-\kappa_4}{\kappa_3-\kappa_1},\quad
b_{1,6}=a_{1,6}\frac{\kappa_3-\kappa_6}{\kappa_3-\kappa_1},\quad
b_{3,4}=a_{3,4}\frac{\kappa_1-\kappa_4}{\kappa_1-\kappa_3},\quad
b_{3,6}=a_{3,6}\frac{\kappa_1-\kappa_6}{\kappa_1-\kappa_3},\quad
\end{align*}
and $b_{2,5}=a_{2,5}, b_{2,7}=a_{2,7}$. Since $A_{1}$ and $A_{2}$ are totally nonnegative, all coefficients $b_{k,j}>0$. This implies that $a_{1,4}, a_{1,6}<0$, while the remaining entries are positive. From these vertex operators, we have the $\tau$-functions given by $\tau_{A_{k}}(x,y,t)=V[A_k]\cdot 1$, i.e.
\begin{align*}
\tau_{A_{1}}(x,y,t)=&
1+b_{1,4}e^{\phi_{1,4}}+b_{1,6}e^{\phi_{1,6}}+b_{3,4}e^{\phi_{3,4}}+b_{3,6}e^{\phi_{3,6}},\\
\tau_{A_{2}}(x,y,t)=&1+b_{2,5}e^{\phi_{2,5}}+b_{2,7}e^{\phi_{2,7}},
\nonumber
\end{align*}
where $\phi_{i,j}=\xi_j(x,y,t)-\xi_i(x,y,t)$. For $\tau_{A_{1}}(x,y,t)$, the terms in the expansion of \(V[A_1]\) corresponding to the same exponential phase $\phi_{1,4}+\phi_{3,6}=\phi_{1,6}+\phi_{3,4}$ combine as
\begin{equation*}
b_{1,4}b_{3,6} X(\kappa_1,\kappa_4) X(\kappa_3,\kappa_6) + b_{1,6}b_{3,4} X(\kappa_1,\kappa_6) X(\kappa_3,\kappa_4)
= (a_{1,4}a_{3,6}-a_{3,4}a_{1,6}) \frac{\kappa_4-\kappa_6}{\kappa_3-\kappa_1} e^{\phi_{1,4}+\phi_{3,6}},
\end{equation*}
which vanishes by construction, reflecting the vanishing of the Pl\"ucker coordinate $\Delta_{46}(A_{1})$.

We now consider the composite $\tau$-function generated by the product of vertex operators
\begin{equation}\label{eq:A1A2}
\tau_{A_{1},A_{2}}(x,y,t)=V[A_1]\cdot V[A_2]\cdot 1.
\end{equation}
Since all the operators $X(\kappa_i,\kappa_j)$ commute, we may write
\begin{align*}
V[A_1]\cdot V[A_2]=\exp\left(b_{1,4}X(\kappa_1,\kappa_4)+b_{1,6}X(\kappa_1,\kappa_6)+
b_{3,4}X(\kappa_3,\kappa_4)\right.\\
\left. +b_{3,6}X(\kappa_3,\kappa_6)+b_{2,5}X(\kappa_2,\kappa_5)+b_{2,7}X(\kappa_2,\kappa_7)\right).\nonumber
\end{align*}
In this product, the cross-ratio factor associated with the pair of operators \(X(\kappa_{1},\kappa_{4})\in V[A_1]\) and \(X(\kappa_{2},\kappa_{5})\in V[A_2]\) is
\begin{align*}
C_{[(1,4),(2,5)]}=\frac{(\kappa_{1}-\kappa_{2})(\kappa_{4}-\kappa_{5})}{(\kappa_{1}-\kappa_{5})(\kappa_{4}-\kappa_{2})}<0.
\end{align*}
Since no other pair of operators corresponds to the same exponential phase, this negative cross-ratio factor produces a negative coefficient in the composite $\tau$-function \eqref{eq:A1A2}, resulting in a singular soliton solution. Using the reordered index sets $\hat{\mathcal{I}}[A_1]$ and $\hat{\mathcal{I}}[A_2]$, the $\kappa$-direct sum is given by
\[
A_{1}\widehat{\bigoplus}A_{2}=\begin{pmatrix}
1& & &a_{1,4} & &a_{1,6} &\\
 & 1& & &a_{2,5} & &a_{2,7}\\
  & &1 & a_{3,4}& &a_{3,6} &\\
\end{pmatrix} \in \Gr(3,7),
\]
where the blank entries are zero. The associated permutation is $\pi(A_{1}\widehat{\bigoplus}A_{2})=(1,3,6,4)(2,7,5)$. The corresponding chord diagram is shown below.
\begin{figure}[H]
  \begin{minipage}[t]{1\linewidth}
  \centering
  \includegraphics[height=1.8cm,width=4.5cm]{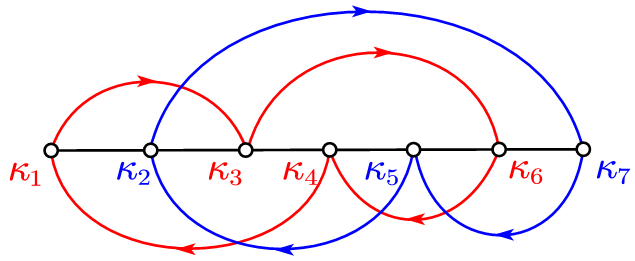}
  \end{minipage}
\end{figure}
\noindent
This loss of total nonnegativity is reflected analytically by negative cross-ratio factors, and combinatorially by crossings between chords from different blocks. In this example, a totally nonnegative realization of this permutation has dimension nine (see Lemma~\ref{lem:chord}), whereas the $\kappa$-direct sum has only five free parameters. Thus the $\kappa$-direct sum cannot remain totally nonnegative.
\end{example}

\begin{remark}
In general, the $\tau$-function generated by the product \(V[A_1]\cdots V[A_K]\) is different from that generated by the single vertex operator associated with the $\kappa$-direct sum. However, the corresponding solutions have the same permutation. To illustrate this difference, let $V[A_{1},A_{2}]$ denote the vertex operator associated with the $\kappa$-direct sum in Example \ref{ex:l5}.
\begin{align*}
V[A_{1},A_{2}]=&\exp\left(c_{1,4}X(\kappa_1,\kappa_4)+c_{1,6}X(\kappa_1,\kappa_6)+c_{2,5}X(\kappa_2,\kappa_5)+\right.\\ \nonumber
&\left.+c_{2,7}X(\kappa_2,\kappa_7)+c_{3,4}X(\kappa_3,\kappa_4)+c_{3,6}X(\kappa_3,\kappa_6)\right).
\end{align*}
Note that $V[A_1,A_2]$ is not equal to $V[A_1]\cdot V[A_2]$, since the coefficients $c_{i,j}$ differ from the $b_{i,j}$. For instance,
\begin{align*}
c_{1,4}=a_{1,4}\frac{(\kappa_2-\kappa_4)(\kappa_3-\kappa_4)}{(\kappa_2-\kappa_1)(\kappa_3-\kappa_1)},\quad\quad c_{2,5}=a_{2,5}\frac{(\kappa_1-\kappa_5)(\kappa_3-\kappa_5)}{(\kappa_1-\kappa_2)(\kappa_3-\kappa_2)}.
\end{align*}
In this case, one must take $a_{1,4}>0$ and $a_{2,5}<0$ to ensure $c_{1,4},c_{2,5}>0$. However, even though all coefficients $c_{i,j}$ are positive, negative terms may still appear in the expansion of $V[A_1,A_2]$ through cross-ratio factors, for example those associated with $X(\kappa_1,\kappa_4)$ and $X(\kappa_2,\kappa_5)$. Therefore the
resulting $\kappa$-direct sum is not totally nonnegative, and the associated solution is also singular.
\end{remark}

We now recall the notion of \emph{crossing} for permutations.
\begin{definition}\label{def:NC}\cite{HY:24}
Two permutations are called \emph{crossing} if their associated chord diagrams contain intersecting chords. Two blocks $A_{i}$ and $A_{j}$ are called \emph{crossing} if the corresponding permutations $\pi(A_{i})$ and $\pi(A_{j})$ are crossing.
\end{definition}
Note that the blocks $A_{1}$ and $A_{2}$ in Example \ref{ex:l5} are crossing. In general, crossings between different blocks give rise to negative Pl\"ucker coordinates of
the $\kappa$-direct sum, leading to singular soliton solutions. This gives the following characterization of regularity.
\begin{proposition}\label{pro:V1K}\cite{HY:24}
Suppose that $A_{k}\in\Gr(n_k,m_k)_{\ge 0}$ for $k=1,\ldots,K$, and let $V[A_k]$ be the vertex operators associated with these blocks. The $\tau$-function \eqref{eq:tau-VO} generated by these vertex operators gives a regular soliton solution \emph{if and only if} the blocks $A_{k}$ are mutually non-crossing.
\end{proposition}

In the next section, we use the Deodhar decomposition to give a combinatorial description of the $\kappa$-direct sum via Go-diagrams.
\section{The Deodhar decomposition of the Grassmannian and Go-diagrams}\label{sec:4}
In this section, we review the Deodhar decomposition of the Grassmannian and its realization via Go-diagrams, which will be used to study the combinatorial structure of the $\kappa$-direct sum.

For a matrix $A\in\Gr(N,M)$ in RREF, let $I=\{i_1,\ldots,i_N\}$ denote its pivot set, which determines a Young diagram inside an $N\times(M-N)$ rectangle. Each box is labeled by a pair $(i,j)$, where $i$ is a pivot index and $j$ is a non-pivot index. This labeling is illustrated in the following diagram.

\bigskip

\setlength{\unitlength}{0.51mm}
\begin{center}
  \begin{picture}(120,60)
\put(5,55){\line(1,0){142}}
\put(5,45){\line(1,0){110}}
  \put(5,35){\line(1,0){90}}
  \put(5,25){\line(1,0){90}}
  \put(5,15){\line(1,0){70}}
  \put(5,5){\line(1,0){70}}
  \put(5,5){\line(0,1){50}}
 \put(31,5){\line(0,1){50}}
  \put(18,5){\line(0,1){50}}
  \put(75,5){\line(0,1){10}}
  \put(62,5){\line(0,1){10}}
  \put(95,25){\line(0,1){10}}
  \put(82,25){\line(0,1){10}}
\put(115,45){\line(0,1){10}}
\put(102,45){\line(0,1){10}}
  \put(98,29){${i_{n}}$}
  \put(-35,29){$M-N+n$}
  \put(-35,49){$M-N+1$}
  \put(-13,38){$\vdots$}
  \put(-13,18){$\vdots$}
  \put(65,38){$\vdots$}
  \put(10,38){$\vdots$}
  \put(24,38){$\vdots$}
  \put(100,38){${}$}
  \put(60,49){$\cdots$}
  \put(60,58){$\cdots$}
  \put(55,29){$\cdots$}
  \put(47,18){$\vdots$}
  \put(10,18){$\vdots$}
  \put(24,18){$\vdots$}
  \put(80,18){${}$}
 \put(-15,9){$M$}
  \put(45,9){$\cdots$}
  \put(78,9){${i_N}$}
  \put(0,58){$M-N$}
 \put(105,58){${i_1}{\qquad\cdots\qquad 1}$}
  \put(118,49){${i_1}{\quad\cdots\quad 1}$}

    \put(7,-3){$ M$}
    \put(18,-3){${M-1\quad\cdots}$}
    \put(60,-3){${i_N+1}$}
 \end{picture}
\end{center}

\medskip

\subsection{Deodhar decomposition}\label{sec:4.1}
We briefly review the Deodhar decomposition of the Grassmannian $\Gr(N,M)$ and the associated \emph{Go-diagrams}. Most of the material in this subsection is based on \cite{Deodhar, MR:04, KW:13}.

We call the Young diagram associated with the pivot set $I=\{i_1,\ldots,i_N\}$ irreducible if
\[
1=i_1<\cdots<i_N<M.
\]
Throughout this section, we assume that the Young diagram is irreducible. Given such a Young diagram contained in an $N \times (M-N)$ rectangle, we label each box by a simple reflection, with the top-right box labeled by $s_1$, and the labels increasing from right to left along each row and by one when passing to the next row. The resulting labeled diagram is denoted by $O_w$. Note that boxes on the same diagonal in $O_w$ carry the same simple reflection, while those whose south-east corner lies on the boundary represent the first occurrences of the simple reflections.

The generators $s_i$, corresponding to the simple reflections $i\leftrightarrow i+1$, satisfy the following commutation and braid relations:
\begin{align*}
s_{i}s_{j}=s_{j}s_{i}\quad\text{if}\quad |i-j|\geq2 \quad\quad\text{and}\quad\quad s_{i}s_{i+1}s_{i}=s_{i+1}s_{i}s_{i+1}.
\end{align*}

For a fixed labeled diagram $O_w$, the boxes are read from bottom-right to top-left according to a \emph{linear extension} of the underlying poset structure; that is, each box is read after all boxes to its right and below. A choice of reading order determines a reduced expression $w=s_{j_1}\cdots s_{j_m}$ for an element $w\in S_M$. We use the symbol $\w$ to emphasize the order of the simple reflections and write $\w=s_{j_1}\cdots s_{j_m}$ for the fixed reading order. Different reading orders give reduced expressions related by the commutation and braid relations, and therefore represent the same element $w$ (see \cite{KW:13} for details).

\begin{example}
Consider the labeled diagram $O_w$ on the left. Two reading orders are shown in the middle and right diagrams, where the numbers indicate the reading order of the boxes.
\begin{figure}[H]
  \begin{minipage}[t]{1\linewidth}
  \centering
 \includegraphics[height=1.4cm,width=8.5cm]{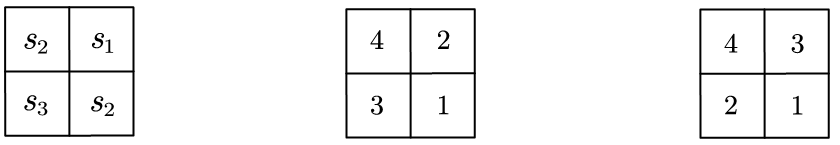}
  \end{minipage}%
\end{figure}
\noindent The corresponding reduced expressions are
\begin{equation*}
\w_{1}=s_{2}s_{1}s_{3}s_{2}, \qquad \w_{2}=s_{2}s_{3}s_{1}s_{2}.
\end{equation*}
\end{example}

Let $\w=s_{j_1}\cdots s_{j_m}$ be a reduced expression. A subexpression $\v$ of $\w$ is obtained by replacing some factors by the identity element $1$, so that
\begin{align*}
\v=v_1\cdots v_m, \qquad v_k \in \{1,s_{j_k}\}.
\end{align*}
For $k\ge 0$, define $v_{(k)}=v_{1}\cdots v_{k}$ with $v_{(0)}=1$. The corresponding group element is given by $v:=v_{(m)}$.

We now recall the notions of distinguished and positive distinguished subexpressions \cite{BB:05, Deodhar, KW:13}. In the following, $\leq$ denotes the Bruhat order on $S_M$, and $<$ denotes its strict version.
\begin{definition}
A subexpression $\v$ of $\w$ is called \emph{distinguished} if
\begin{equation}\label{eq:dist}
v_{(k)}\le v_{(k-1)}s_{j_k}\quad \text{for all}\quad k=1,\ldots,m.
\end{equation}
Equivalently, in terms of the length function on $S_M$, whenever right multiplication by $s_{j_k}$ decreases the length of $v_{(k-1)}$, the factor $s_{j_k}$ must be chosen at position $k$ in $\v$. A subexpression $\v$ is called a positive distinguished subexpression (PDS) of $\w$ if
\begin{equation}\label{eq:pds}
v_{(k-1)}< v_{(k-1)}s_{j_k} \quad \text{for all}\quad k=1,\ldots,m.
\end{equation}
Equivalently, it is a distinguished subexpression whose length never decreases. We write $\v\prec\w$ for a distinguished subexpression, and call $\v$ a PDS if it satisfies \eqref{eq:pds}.
\end{definition}

\begin{remark}\cite{KW:13}\label{re:v}
Whether a subexpression $\v$ is distinguished is independent of the choice of reading order.
\end{remark}

Given a subexpression $\v$ of $\w$, we associate to each position $k$ (equivalently, to each box of $O_{w}$ under any reading order) a decoration as follows:
\begin{itemize}
\item place a \raisebox{0.12cm}{\hskip0.14cm\circle{5.5}\hskip-0.1cm} (white stone) if $v_{(k-1)}<v_{(k)}$;
\item place a \raisebox{0.12cm}{\hskip0.14cm\circle*{5.5}\hskip-0.1cm} (black stone) if $v_{(k-1)}>v_{(k)}$;
\item leave the box empty if $v_{(k-1)}=v_{(k)}$.
\end{itemize}
In this way, empty boxes correspond to positions where $v_k=1$, while nonempty boxes correspond to positions where $v_k=s_{j_k}$, with the color determined by whether the
corresponding step is length-increasing or length-decreasing. If $\v$ is distinguished, the resulting decorated Young diagram is called a \emph{Go-diagram}. In this setting,
the failure of total nonnegativity in the Grassmannian is detected by length-decreasing steps, which appear as black stones in the diagram.
\begin{remark}
A $\Le$-diagram, introduced by Postnikov \cite{P:06}, can be viewed as a Go-diagram without black stones. If $\v$ is a PDS of $\w$, then the associated Go-diagram contains no black stones and reduces to a $\Le$-diagram. Such $\Le$-diagrams parametrize positroid cells in the totally nonnegative Grassmannian $\Gr(N,M)_{\ge0}$. In particular, a $\Le$-diagram is characterized by the $\Le$-property: if a box contains a \raisebox{0.12cm}{\hskip0.14cm\circle{5.5}\hskip-0.1cm}, then all boxes to its left or above also contain \raisebox{0.12cm}{\hskip0.14cm\circle{5.5}\hskip-0.1cm} (see \cite{KW:13} for details).
\end{remark}

\begin{remark}
Since boxes on the same diagonal carry the same simple reflection, a black stone is naturally paired with a white stone on the same diagonal, although this diagonal correspondence may be altered by braid relations (see Example~\ref{ex:hirota} below).
\end{remark}
\begin{example}\label{ex:w4}
Let $\w=s_{2}s_{3}s_{1}s_{2}$. We consider the following subexpressions.
\begin{itemize}
\item[(a)] $\v=s_{2}111$ is not a distinguished subexpression, since $v_{(4)}=s_{2}>v_{(3)}s_{2}=1$.
\item[(b)] $\v=1s_{3}s_{1}1$ is a positive distinguished subexpression. Indeed,
\[
v_{(k-1)} < v_{(k-1)} s_{j_k} \qquad \text{for all } k=1,\dots,4.
\]
\item[(c)] $\v=s_{2}11s_{2}$ is a distinguished subexpression, since $v_{(1)}=v_{(2)}=v_{(3)}=s_{2}>v_{(4)}=v_{(3)}s_{2}=1$.
\end{itemize}
The associated diagrams are shown below.
\begin{figure}[H]
  \begin{minipage}[t]{1\linewidth}
  \centering
 \includegraphics[height=1.8cm,width=8.7cm]{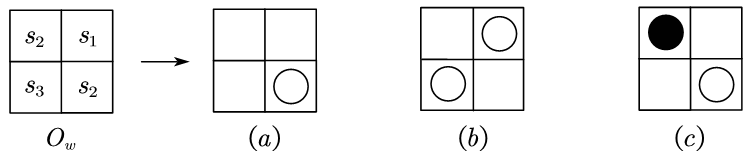}
  \end{minipage}%
\end{figure}
\noindent Note that neither diagram (a) nor (b) contains black stones. However, diagram (a) does not satisfy the $\Le$-property and therefore does not correspond to a $\Le$-diagram. In contrast, diagram (b) arises from a PDS and hence is a $\Le$-diagram, while diagram (c) represents a Go-diagram containing a black stone.
\end{example}

\subsection{Deodhar components and the associated permutations}\label{sec:projections}
Recall from \cite{KW:13} that the positroid cell decomposition of the irreducible totally nonnegative Grassmannian is given by
\begin{equation*}
\Gr(N,M)_{\ge0}^{\mathrm{irr}}=\bigsqcup_{\pi\in D_M}X_\pi,
\end{equation*}
where $D_M\subset S_M$ denotes the set of irreducible derangements on $[M]$. Each positroid cell is uniquely determined by the permutation $\pi$.

More generally, the Deodhar decomposition of the complete flag variety refines the decomposition into Richardson varieties, which are intersections of opposite Bruhat cells. Projecting this decomposition to the Grassmannian gives a refinement of the positroid stratification~\cite{KW:13}.
Each distinguished subexpression $\v\prec\w$ determines a Deodhar component labeled by $(\v,\w)$. For a Grassmannian element belonging to the component labeled by $(\v,\w)$, one obtains a permutation
\begin{equation*}
\pi=vw^{-1},
\end{equation*}
where $v=v_{(m)}$ and $w$ are the group elements represented by $\v$ and $\w$, respectively \cite{MR:04, KW:13}. In particular, if $\v$ is a PDS, then $\pi$ agrees with the permutation indexing the corresponding positroid cell.

Continuing with case $(c)$ of Example \ref{ex:w4}, and viewing $s_i$ as the transposition $i\leftrightarrow i+1$, we obtain
\begin{align*}
\pi=vw^{-1}=(1,3)(2,4).
\end{align*}

\begin{remark}
The permutation $\pi$ determines the underlying Young diagram but does not \emph{uniquely} determine the Go-diagram, since different distinguished subexpressions may yield the same permutation.
\end{remark}
\begin{example}\label{ex:s24}
Fix $\w=s_{2}s_{3}s_{1}s_{2}$ and consider the two distinguished subexpressions
\begin{align*}
\v_{1}=s_{2}11s_{2},\quad\quad\quad\v_{2}=1111.
\end{align*}
Since $v=1$ for both subexpressions, we have $\pi=w^{-1}$. However, the two subexpressions give rise to different Go-diagrams: $\v_{1}$ contains a length-decreasing step, whereas $\v_{2}$ is a PDS (so the associated diagram has no black stones and reduces to a $\Le$-diagram).
\end{example}

\begin{remark}
Although the permutation alone does not uniquely determine the Go-diagram,
the matrix $A$ contains sufficient information to recover it. In particular,
for a soliton solution generated by a general vertex operator $V[A]$, the
associated Go-diagram can be determined from the Pl\"ucker coordinates of
$A$ (see \cite[Remark~7.11]{TW:13} for the details).
\end{remark}

\subsection{Composite diagrams of the $\kappa$-direct sum}
We now apply the Deodhar decomposition to the $\kappa$-direct sum. The $\kappa$-direct sum $\widehat{\bigoplus}_{k=1}^K A_k$ is constructed from blocks associated with vertex
operators $V[A_k]$ (see Eq.~\eqref{eq:VO}). Although each block $A_k$ is totally nonnegative, the $\kappa$-direct sum may not be, since cross-ratio factors arising between different blocks may change the signs of the Pl\"ucker coordinates of the $\kappa$-direct sum, leading to singular soliton solutions.

Each totally nonnegative block $A_k\in\Gr(n_k,m_k)_{\ge 0}$ corresponds to a positive distinguished subexpression $\v_k \prec \w_k$, which determines a $\Le$-diagram $\mathcal{L}_{\v_k,\w_k}$.

The permutation of the $\kappa$-direct sum is given by \eqref{eq:kapp}, i.e.
\begin{equation*}
\pi(\widehat{\bigoplus}_{k=1}^K A_k)=\prod_{k=1}^{K}\pi(A_{k}).
\end{equation*}
This permutation determines the underlying Young diagram. Each block $A_k$ contributes a $\Le$-diagram $\mathcal{L}_{\v_k,\w_k}$. It remains to determine how these diagrams fit together.

We now introduce a composite diagram related to the $\kappa$-direct sum $\widehat{\bigoplus}_{k=1}^K A_k\in\Gr(N,M)$, where $N=n_1+\cdots+n_K$ and $M=m_1+\cdots+m_K$.
\begin{definition}\label{def:composite}
The composite diagram
\begin{align*}
\widehat{\bigoplus}_{k=1}^K\mathcal{L}_{\v_{k},\w_{k}}
\end{align*}
is obtained by embedding each $\Le$-diagram $\mathcal{L}_{\v_{k},\w_{k}}$ into the global Young diagram at the positions determined by the reordered index sets $\hat{\mathcal I}[A_k]$. The remaining boxes are filled with nonempty \emph{labeled stones}, assigned from right to left along each row, proceeding from the bottom row upward (see Example~\ref{ex:cd}).
\end{definition}

We show that the composite diagram determines a distinguished subexpression. Fix a reading order of the global Young diagram, and let $\w^{\oplus}$ be the corresponding reduced expression. The filling of the composite diagram determines a subexpression $\v^{\oplus}$ of $\w^{\oplus}$.
\begin{proposition}\label{pro:dis}
The subexpression $\v^{\oplus}$ determined by the composite diagram is distinguished.
\end{proposition}
\begin{proof}
By the characterization of distinguished subexpressions, it suffices to show that for each $k$, whenever
\begin{equation*}
v^{\oplus}_{(k-1)}s_{j_k}<v^{\oplus}_{(k-1)},
\end{equation*}
the corresponding simple reflection $s_{j_k}$ is chosen in $\v^{\oplus}$. Since each embedded $\Le$-diagram is preserved in the composite diagram and corresponds to a PDS,
all steps coming from the embedded $\Le$-diagrams are length-increasing. Any length-decreasing step must occur in the additional boxes, which are nonempty by definition
of the composite diagram. Thus the corresponding simple reflection must be chosen in $\v^{\oplus}$.
\end{proof}
\begin{example}\label{ex:cd}
We continue Example \ref{ex:l5}. The $\Le$-diagrams $\mathcal{L}_{\v_{1},\w_{1}}$ and $\mathcal{L}_{\v_{2},\w_{2}}$ corresponding to $A_{1}$ and $A_{2}$ are shown on the left.
The reordered index sets are $\hat{\mathcal I}[A_1]=\{1,3,4,6\}$ and $\hat{\mathcal I}[A_2]=\{2,5,7\}$. The resulting composite diagram is shown in Figure \ref{fig:v1v2}.
\begin{figure}[H]
  \begin{minipage}[t]{1\linewidth}
  \centering
 \includegraphics[height=1.7cm,width=9.5cm]{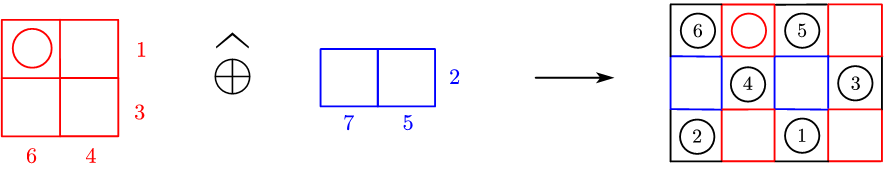}
  \end{minipage}%
  \caption{The composite diagram $\mathcal{L}_{\v_{1},\w_{1}}\widehat{\bigoplus}\mathcal{L}_{\v_{2},\w_{2}}$.}\label{fig:v1v2}
\end{figure}
\noindent We fix a reading order in which each row is read from right to left, proceeding row by row from bottom to top. The corresponding reduced expressions and positive
distinguished subexpressions are
\begin{align*}
\w_{1}=s_{2}s_{3}s_{1}s_{2}, \quad \v_{1}=111s_{2}
\quad\text{and}\quad
\w_{2}=s_{1}s_{2}, \quad \v_{2}=11.
\end{align*}
Applying the same reading order to the composite diagram, we have
\begin{equation*}
\w^{\oplus}=s_{3}s_{4}s_{5}s_{6}s_{2}s_{3}s_{4}s_{5}s_{1}s_{2}s_{3}s_{4},\quad\quad \v^{\oplus}=1s_{4}1s_{6}s_{2}1s_{4}11s_{2}s_{3}s_{4}.
\end{equation*}
\end{example}

To determine whether each labeled stone is either white or black, we introduce the associated in-region and out-region. More precisely, the type of each labeled stone is determined by a suitable in-region and can be identified without computing the entire subexpression $\v^{\oplus}$.

\begin{definition}\label{def:inout}
For any box $b$ in the composite diagram, we define the in-region $Y_b^{\In}$ to be the set of boxes weakly south-east of $b$, and the out-region $Y_b^{\Out}$ to be its complement.
\end{definition}

Using the labeling introduced at the beginning of this section, we write $b=(i,j)$ for the box labeled by $(i,j)$. One can always choose a reading order in which all boxes in $Y_b^{\In}$ are read before those in $Y_b^{\Out}$ \cite{KW:13}. This gives a factorization
\begin{equation*}
\w^{\oplus}=\w^{\In}\w^{\Out}, \qquad \v^{\oplus}=\v^{\In}\v^{\Out}.
\end{equation*}
Since the distinguished structure of $\v^{\oplus}$ is independent of the chosen reading order (see Remark \ref{re:v}), the type of the labeled stones in the in-region $Y_b^{\In}$ is determined by the subexpression $\v^{\In}$. This provides a convenient way to distinguish white and black stones, which in the PDS case is characterized combinatorially by the $\Le$-property.
\begin{lemma}\label{lem:Le-in-region}
Let $b$ be a box in the composite diagram. If the filling restricted to the in-region $Y_b^{\In}$ satisfies the $\Le$-property, then all labeled stones in $Y_b^{\In}$ are white. In this case, if $b$ is labeled, then it is white.
\end{lemma}
\begin{proof}
Since $\v^\oplus$ is distinguished, the subexpression $\v^{\rm in}$ associated with $Y_b^{\In}$ is also distinguished. If the filling restricted to $Y_b^{\In}$ satisfies the $\Le$-property, then $\v^{\rm in}$ is a PDS and contains no length-decreasing steps. Thus all labeled stones in $Y_b^{\In}$ are white.
\end{proof}
The filling of the composite diagram can be described in terms of nested in-regions. As new boundary layers are added, the filling on the previous region remains unchanged,
and only the new boxes need to be determined. As long as the $\Le$-property is preserved, the new labeled stones are white; black stones appear precisely when the $\Le$-property fails for the first time in the added part.
\begin{example}
We illustrate the filling process using the composite diagram in Figure \ref{fig:v1v2}. Consider the in-region associated with $b=(2,5)$. The restriction of the filling to $Y_{(2,5)}^{\In}$ satisfies the $\Le$-property, and hence the labeled stones \tikz[baseline=(char.base)]{
  \node[draw, circle, minimum size=1.0em, inner sep=0pt] (char) {\small 1};} and \tikz[baseline=(char.base)]{
  \node[draw, circle, minimum size=1.0em, inner sep=0pt] (char) {\small 3};} are white.

Consider next the in-region associated with $b=(1,5)$, obtained by extending $Y_{(2,5)}^{\In}$ upward. Here the $\Le$-property fails for the first time at the newly added labeled stone \tikz[baseline=(char.base)]{ \node[draw, circle, minimum size=1.0em, inner sep=0pt] (char) {\small 5};}, which is therefore black.
\begin{figure}[H]
  \begin{minipage}[t]{1\linewidth}
  \centering
 \includegraphics[height=4cm,width=9.2cm]{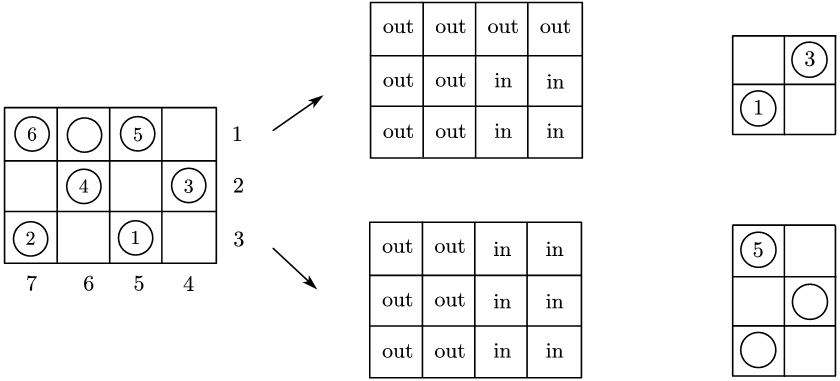}
\caption{The in-regions associated with $b=(2,5)$ (top) and $b=(1,5)$ (bottom).}
  \end{minipage}%
\end{figure}
\noindent This example illustrates the transition from regions satisfying the $\Le$-property to its first violation.
\end{example}

\begin{remark}
The in-region description identifies the portion of the diagram relevant to a labeled stone. To determine the type of successive labeled stones, it is convenient to use the pipe dream representation, where their types can be identified by tracking the corresponding pipes. This viewpoint will be developed in the next subsection.
\end{remark}

\subsection{Go-diagrams via pipe dream}\label{sec4.4}
We now use the \emph{pipe dream} representation to determine the type of the labeled stones.

By Proposition~\ref{pro:dis}, the subexpression $\v^{\oplus}$ determined by the composite diagram is distinguished. The composite diagram therefore defines a Go-diagram in the Deodhar decomposition and admits a pipe dream representation.

This representation is constructed as follows: empty boxes correspond to elbows with bridges, while nonempty boxes correspond to crossings (see \cite{K:17, KW:13}). Boxes carrying labeled stones correspond to crossings. Their types will be determined below by tracking the corresponding pipes.
\begin{figure}[H]
  \begin{minipage}[t]{1\linewidth}
  \centering
 \includegraphics[height=0.7cm,width=6.0cm]{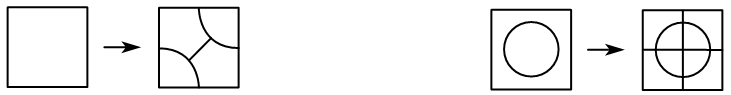}
  \end{minipage}%
\end{figure}
We label the $M$ edges along the south-east boundary by $1,\ldots,M$, from north-east to south-west.
Each label is carried along the corresponding pipe: at a crossing the pipe continues straight, while at an elbow it turns. Thus every edge segment of a pipe carries the label of its starting edge. The permutation $\pi$ is obtained by tracking each pipe from $i$ on the south-east boundary to $\pi(i)$ on the north-west boundary.

The relation between this permutation and the factorization in \eqref{eq:kapp} is described in the following lemma.
\begin{lemma}\label{lem:pipe}
The pipe dream associated with each embedded $\Le$-diagram is preserved in the composite diagram. Moreover, the permutation determined by the resulting pipe dream is
\[
\pi=\prod_{k=1}^{K}\pi(A_k).
\]
\end{lemma}
\begin{proof}
The filling of each embedded $\Le$-diagram is unchanged, while the additional boxes correspond to crossing tiles and therefore do not change the connectivity of the pipes. Thus each block contributes the permutation $\pi(A_k)$. Since the reordered index sets $\hat{\mathcal I}[A_k]$ are disjoint, the permutation determined by the resulting pipe dream can be expressed as the product of the block permutations $\pi(A_k)$. This is precisely the permutation formula for the $\kappa$-direct sum in \eqref{eq:kapp}, obtained in \cite{HY:24}.
\end{proof}

An illustration of Lemma~\ref{lem:pipe} for the composite diagram in Figure~\ref{fig:v1v2} is shown in Figure~\ref{fig:pd}. In this example, tracking the pipes gives the permutation
$\pi=(1,3,6,4)(2,7,5)$.
\begin{figure}[H]
  \begin{minipage}[t]{1\linewidth}
  \centering
 \includegraphics[height=2.5cm,width=10cm]{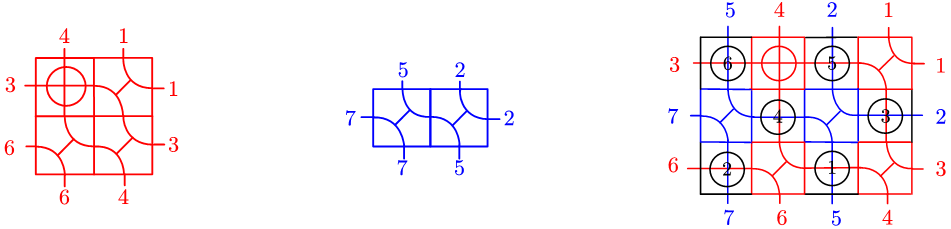}
\caption{Pipe dream associated with the composite diagram in Figure~\ref{fig:v1v2}.}\label{fig:pd}
  \end{minipage}%
\end{figure}

Given a labeled stone in the composite diagram, we consider the two pipes passing through the corresponding box in the pipe dream. Each pipe connects a boundary label $i$ on the south-east boundary to the corresponding label $\pi(i)$ on the north-west boundary, thereby determining a line-soliton labeled by $[i,\pi(i)]$, or simply by $[i,j]$ with $\pi(i)=j$.

As illustrated in Figure~\ref{fig:sol}, a crossing box determines two line-solitons $[i,j]$ and $[k,l]$, where $\pi(i)=j$ and $\pi(l)=k$. Without loss of generality, we assume that $i<k$. This pairing determines the type of the labeled stone.
\begin{figure}[H]
  \begin{minipage}[t]{1\linewidth}
  \centering
 \includegraphics[height=3.3cm,width=13cm]{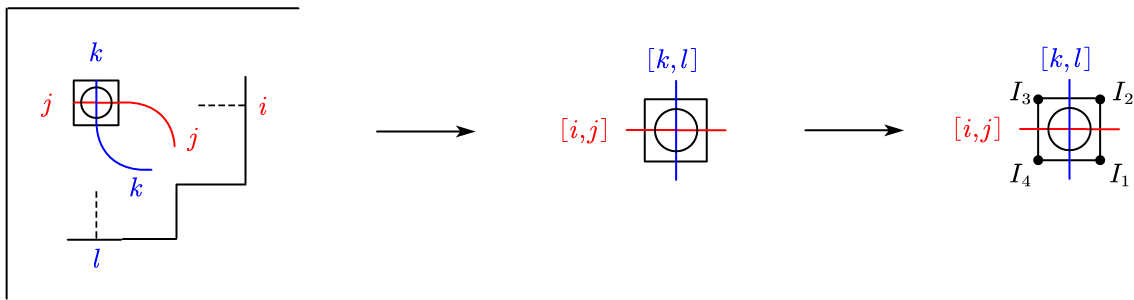}
 \caption{Index pairing in the pipe dream.}\label{fig:sol}
  \end{minipage}
\end{figure}
\noindent We call such a crossing box an $X$-crossing if the associated pair of line-solitons $[i,j]$ and $[k,l]$ has pairwise distinct indices $i,j,k,l$. We denote by $\mathcal{B}$ the set of $X$-crossings carrying labeled stones. For each box $b\in\mathcal{B}$, the associated line-solitons can be read off from the pipe configuration restricted to the in-region $Y_b^{\In}$, which is the relevant region for determining the type of the labeled stone.

Notice that an $X$-crossing gives rise to four adjacent dominant exponential regions with associated index sets $I_1,I_2,I_3,I_4$, arranged counterclockwise around the $X$-crossing. Its type is determined by the relative positions of the associated line-soliton pairs in the chord diagram. We have the following classification.
\begin{definition}\label{def:OTP}
Let $[i,j]$ and $[k,l]$ be two line-solitons with $i<k$. Their relative positions in the chord diagram are classified as follows:
\begin{itemize}
\item[(a)] $P$-type if $i<k<l<j$, i.e. one pair completely contains the other;
\item[(b)] $T$-type if $i<k<j<l$, i.e. the two pairs partially overlap;
\item[(c)] $O$-type if $i<j<k<l$, i.e. the two pairs are disjoint.
\end{itemize}
\begin{figure}[H]
  \begin{minipage}[t]{1\linewidth}
  \centering
 \includegraphics[height=1.8cm,width=10cm]{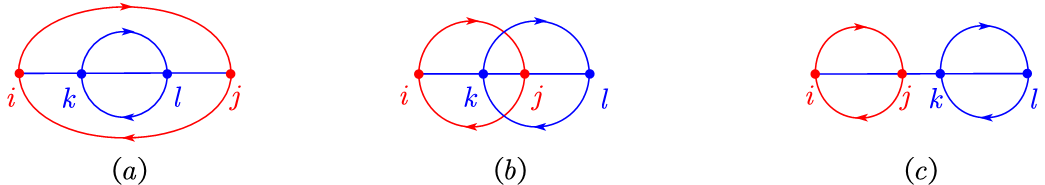}
  \end{minipage}%
  \caption{Chord diagrams illustrating these three types.}
\end{figure}
\end{definition}

\begin{remark}
For boxes in $\mathcal{B}$, the boundary ordering of the corresponding pipes is non-nesting, as illustrated in Figure~\ref{fig:sol}. In particular, the associated line-soliton pairs cannot satisfy the nested ordering $i<k<l<j$, so that only $O$- and $T$-types can occur for boxes in $\mathcal{B}$.
\end{remark}

The following theorem shows that the type of an $X$-crossing determines the corresponding sign relation among the Pl\"ucker coordinates.
\begin{theorem}\label{th:PC}
Let $A=\widehat{\bigoplus}_{k=1}^K A_k\in \Gr(N,M)$, and fix an $X$-crossing in its composite diagram. Let $I_1,I_2,I_3,I_4$ be the index sets of the four dominant exponential regions adjacent to the $X$-crossing, arranged counterclockwise. The following relations hold.
\begin{itemize}
\item[(a)] If the $X$-crossing is of $O$-type, then $\Delta_{I_{1}}(A)\Delta_{I_{3}}(A)=\Delta_{I_{2}}(A)\Delta_{I_{4}}(A)$.
\item[(b)] If the $X$-crossing is of $T$-type, then $\Delta_{I_{1}}(A)\Delta_{I_{3}}(A)=-\Delta_{I_{2}}(A)\Delta_{I_{4}}(A)$.
\end{itemize}
\end{theorem}
\begin{proof}
The four regions adjacent to the given $X$-crossing correspond to index sets $I_1,I_2,I_3,I_4$ differing only in the indices $i,j,k,l$. Let $J$ be their common $(N-2)$-element subset. The index sets are given by
\begin{align*}
I_{1}=J\cup\{i,k\},\quad I_{2}=J\cup\{k,j\},\quad I_{3}=J\cup\{j,l\},\quad I_{4}=J\cup\{i,l\}.
\end{align*}
Here we use the fact that adjacent regions differ by a single index, since each line-soliton arises from the balance between two dominant exponentials whose index sets differ by exactly one element. Applying the Pl\"ucker relation to the four indices $i,j,k,l$ with $i<j<k<l$, we obtain
\begin{align*}
\Delta_{J\cup\{i,j\}}(A)\Delta_{J\cup\{k,l\}}(A)-\Delta_{J\cup\{i,k\}}(A)\Delta_{J\cup\{j,l\}}(A)
+\Delta_{J\cup\{i,l\}}(A)\Delta_{J\cup\{j,k\}}(A)=0.
\end{align*}
For the ordering $i<k<j<l$, the first two products are interchanged. In both orderings, the term $\Delta_{J\cup\{i,j\}}(A)\Delta_{J\cup\{k,l\}}(A)$ corresponds to the pairing not realized by the given $X$-crossing. By the Pl\"ucker-coordinate relations for the labeled stone corresponding to the given $X$-crossing
\cite[Corollary~5.8 and Remark~5.9]{KW:13}, the above term is the vanishing term in the three term Pl\"ucker relation, i.e., $\Delta_{J\cup\{i,j\}}(A)\Delta_{J\cup\{k,l\}}(A)=0$. Thus the sign relation associated with the crossing is determined by the remaining two terms. If the $X$-crossing is of $O$-type, we have
\begin{align*}
\Delta_{I_1}(A)\Delta_{I_3}(A)=\Delta_{I_2}(A)\Delta_{I_4}(A).
\end{align*}
If the $X$-crossing is of $T$-type, the interchange of the first two terms changes the sign and gives
\begin{align*}
\Delta_{I_1}(A)\Delta_{I_3}(A)=-\Delta_{I_2}(A)\Delta_{I_4}(A).
\end{align*}
\end{proof}

\begin{remark}
An $X$-crossing in the pipe dream determines a pair of line-solitons that intersect in the corresponding soliton graph for $t\to -\infty$ \cite{KW:13}. This correspondence becomes more transparent after reducing the pipe dream by first eliminating the edges carrying only one index and then removing the intermediate vertices incident to exactly two edges (see Chapter~3 in the book \cite{K:17}).
\end{remark}

For the $\kappa$-direct sum, the asymptotic soliton graph has a more restricted structure. Line-solitons arising from different blocks may intersect, but such intersections do not produce the opening of a box at the intersection point (see Section~6.4 of \cite{CK:09}). This is due to the fact that no additional free parameters arise from the combination of different blocks. Moreover, according to Theorem~\ref{th:PC}, the type of the corresponding $X$-crossing determines the sign relation among the adjacent dominant exponential regions, thereby determining whether the intersecting line-solitons preserve or change their regularity when passing through the intersection.
\begin{example}\label{ex:reducepd}
We illustrate the $X$-crossings arising in the composite diagram for the $\kappa$-direct sum $A_{1}\widehat{\bigoplus}A_{2}$ from Example~\ref{ex:l5}. Figure~\ref{fig:CP37} shows the corresponding reduced pipe dream. The associated soliton graph is displayed in Figure~\ref{fig:Gr37} for comparison.
\begin{figure}[H]
  \begin{minipage}[t]{1\linewidth}
  \centering
  \includegraphics[height=2.5cm,width=10cm]{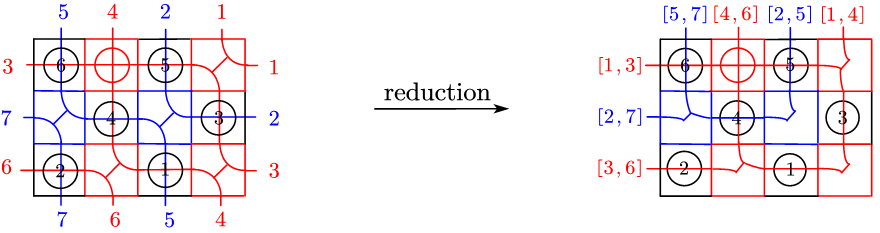}
  \caption{Pipe dream associated with the composite diagram and its reduced form.}\label{fig:CP37}
  \end{minipage}
  \end{figure}
  \begin{figure}[H]
  \begin{minipage}[t]{1\linewidth}
  \centering
 \includegraphics[height=3.4cm,width=11cm]{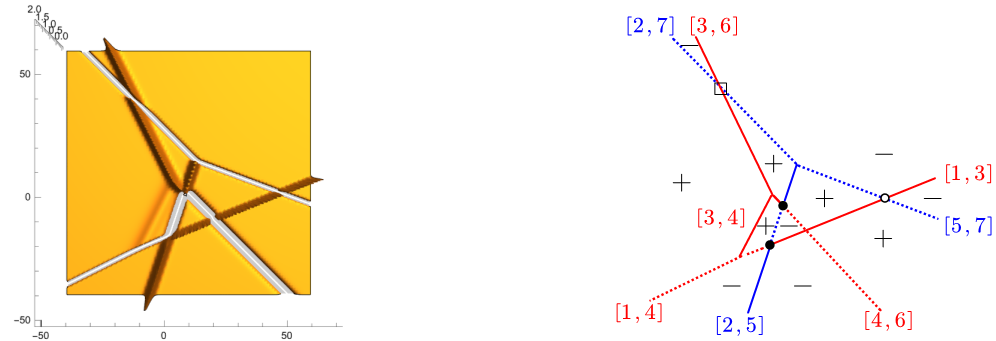}
  \end{minipage}%
\caption{Soliton graph of the solution associated with $A_{1}\widehat{\bigoplus}A_{2}$ at $t=-10$. The signs ``$+$'' and ``$-$'' assigned to each region indicate the sign of the Pl\"ucker coordinate corresponding to the dominant exponential in that region.\label{fig:Gr37}}
\end{figure}
\noindent We take $(\kappa_{1},\ldots,\kappa_{7})=(-2,-1,-\tfrac{1}{2},0,\tfrac{2}{3},1,2)$. In this example, the labeled stones \tikz[baseline=(char.base)]{
  \node[draw, circle, minimum size=1.0em, inner sep=0pt] (char) {\small 4};} \tikz[baseline=(char.base)]{
  \node[draw, circle, minimum size=1.0em, inner sep=0pt] (char) {\small 5};} \tikz[baseline=(char.base)]{
  \node[draw, circle, minimum size=1.0em, inner sep=0pt] (char) {\small 6};} correspond to boxes in $\mathcal{B}$.
From the reduced pipe dream in Figure~\ref{fig:CP37}, we can read off the pairs of line-solitons associated with the labeled stones. The corresponding $X$-crossings are reflected in the
soliton graph shown in Figure~\ref{fig:Gr37}. Solid lines denote regular line-solitons (adjacent dominant exponentials with the same sign) and dashed lines denote singular solitons (opposite signs). This example contains $X$-crossings of both $T$- and $O$-types. For instance:
\begin{itemize}
\item[(a)] The $X$-crossing formed by the $[2,5]$- and $[4,6]$-solitons of $T$-type (black circle in Figure \ref{fig:Gr37}) gives
\begin{align*}
\Delta_{I_{1}}=\Delta_{126}=a_{3,6},\quad \Delta_{I_{2}}=\Delta_{124}=a_{3,4},\quad \Delta_{I_{3}}=\Delta_{145}=-a_{2,5}a_{3,4},\quad \Delta_{I_{4}}=\Delta_{156}=a_{2,5}a_{3,6}.
\end{align*}
\item[(b)] The $X$-crossing formed by the $[1,3]$- and $[5,7]$-solitons of $O$-type (white circle in Figure \ref{fig:Gr37}) gives
\begin{align*}
\Delta_{I_{1}}=\Delta_{156}=a_{2,5}a_{3,6},\quad \Delta_{I_{2}}=\Delta_{356}=-a_{1,6}a_{2,5},\quad \Delta_{I_{3}}=\Delta_{367}=a_{1,6}a_{2,7},\quad \Delta_{I_{4}}=\Delta_{167}=-a_{2,7}a_{3,6}.
\end{align*}
\end{itemize}
These computations verify the sign relations in Theorem~\ref{th:PC}.
\end{example}

\begin{remark}
The composite diagram records only those $X$-crossings corresponding to boxes in $\mathcal B$, arising from intersections of pipes associated with different blocks. Thus,
although two line-solitons may intersect in the soliton graph with index ordering $i<k<l<j$ (such as the square vertex in Figure~\ref{fig:Gr37}) for certain values of $t$,
they do not correspond to elements of $\mathcal{B}$.
\end{remark}

Combining the Pl\"ucker relations (see, e.g. \cite[Corollary 5.8]{KW:13}) with Theorem~\ref{th:PC}, we obtain the following characterization of labeled stone types.
\begin{corollary}\label{cor:wb}
For each box $b\in\mathcal{B}$, the filling of the labeled stone is determined by the type of the corresponding $X$-crossing. In particular,
\begin{itemize}
\item[(a)] $O$-type $X$-crossings correspond to white stones;
\item[(b)] $T$-type $X$-crossings correspond to black stones.
\end{itemize}
\end{corollary}
We now consider the labeled stones not belonging to $\mathcal{B}$. In Example~\ref{ex:reducepd}, these are precisely the labeled stones \tikz[baseline=(char.base)]{
\node[draw, circle, minimum size=1.0em, inner sep=0pt] (char) {\small 1};} \tikz[baseline=(char.base)]{
\node[draw, circle, minimum size=1.0em, inner sep=0pt] (char) {\small 2};} \tikz[baseline=(char.base)]{
\node[draw, circle, minimum size=1.0em, inner sep=0pt] (char) {\small 3};}.

\begin{proposition}\label{prop:boundary}
For labeled stones not belonging to $\mathcal{B}$, the associated configuration is degenerate: some of the indices $i,j,k,l$ coincide and hence some of the index sets $I_1,I_2,I_3,I_4$ coincide. In this case, the Pl\"ucker relation associated with this configuration reduces to
\[
\Delta_{I_1}(A)\Delta_{I_3}(A)
=
\Delta_{I_2}(A)\Delta_{I_4}(A).
\]
This proves that all such labeled stones are white.
\end{proposition}

Combining the preceding results, we obtain the main theorem of this section.
\begin{theorem}\label{th:main}
The composite diagram associated with the $\kappa$-direct sum admits a \emph{canonical} filling into a Go-diagram, determined by explicit local rules. In particular,
\begin{itemize}
\item boxes inside the embedded $\Le$-diagrams remain unchanged;
\item the remaining boxes in $\mathcal{B}$ are filled according to Corollary~\ref{cor:wb};
\item boxes not belonging to $\mathcal B$ are filled with white stones.
\end{itemize}
This provides an explicit realization of the Go-diagram associated with the $\kappa$-direct sum.
\end{theorem}

By Definition~\ref{def:composite} and Theorem~\ref{th:main}, we obtain the following characterization.
\begin{corollary}\label{cor:kappa-Go}
A Go-diagram arises from a $\kappa$-direct sum $\widehat{\bigoplus}_{k=1}^K A_k\in\Gr(N,M)$ if and only if it admits a decomposition into embedded $\Le$-diagrams whose corresponding reordered index sets $\widehat{\mathcal I}[A_k]$ form a partition of $[M]$, with all additional boxes nonempty.
\end{corollary}

\begin{corollary}\label{cor:Le}
Let $A=\widehat{\bigoplus}_{k=1}^K A_k$. Then the following conditions are equivalent:
\begin{enumerate}[(i)]
\item the composite diagram $\widehat{\bigoplus}_{k=1}^K\mathcal{L}_{\v_{k},\w_{k}}$ is a $\Le$-diagram;
\item no $T$-type $X$-crossings occur in $\mathcal{B}$;
\item the blocks $A_k$ are mutually non-crossing.
\end{enumerate}
\end{corollary}
\begin{proof}
The equivalence of (i) and (ii) follows from Corollary~\ref{cor:wb}, since black stones occur precisely at $T$-type $X$-crossings. Proposition~\ref{pro:V1K} implies that condition (iii) is equivalent to the regularity of the corresponding soliton solution. Since regular soliton solutions are parametrized by $\Le$-diagrams, condition (iii) is equivalent to (i).
\end{proof}
This characterization applies in particular to KP soliton solutions under the Gel'fand--Dickey reductions constructed via vertex operators \cite{HY:24}.

\section{Examples}\label{sec:5}
In this section, we illustrate the construction of the Go-diagram associated with the $\kappa$-direct sum through several examples.
\subsection{A case of $\text{Gr}(3,7)$}
\begin{example}\label{ex:2me}
Continuing Example~\ref{ex:l5}, we consider the composite diagram associated with the $\kappa$-direct sum $A_{1}\widehat{\bigoplus}A_{2}$. From the pipe dream, we identify the following $X$-crossings:
\begin{itemize}
\item[{\tikz[baseline=(char.base)]{
  \node[draw, circle, minimum size=1.0em, inner sep=0pt] (char) {\small 4};
} :}] the $T$-type $X$-crossing formed by the $[2,5]$- and $[4,6]$-solitons;
\item[{\tikz[baseline=(char.base)]{
  \node[draw, circle, minimum size=1.0em, inner sep=0pt] (char) {\small 5};
} :}] the $T$-type $X$-crossing formed by the $[1,3]$- and $[2,5]$-solitons;
\item[{\tikz[baseline=(char.base)]{
  \node[draw, circle, minimum size=1.0em, inner sep=0pt] (char) {\small 6};
} :}] the $O$-type $X$-crossing formed by the $[1,3]$- and $[5,7]$-solitons.
\end{itemize}
By Corollary~\ref{cor:wb}, the labeled stones \tikz[baseline=(char.base)]{
\node[draw, circle, minimum size=1.0em, inner sep=0pt] (char) {\small 4};} and \tikz[baseline=(char.base)]{
\node[draw, circle, minimum size=1.0em, inner sep=0pt] (char) {\small 5};} are black, while all others are white. This determines the Go-diagram associated with $A_{1}\widehat{\bigoplus}A_{2}$.
\begin{figure}[H]
  \begin{minipage}[t]{1\linewidth}
  \centering
  \includegraphics[height=1.8cm,width=8.5cm]{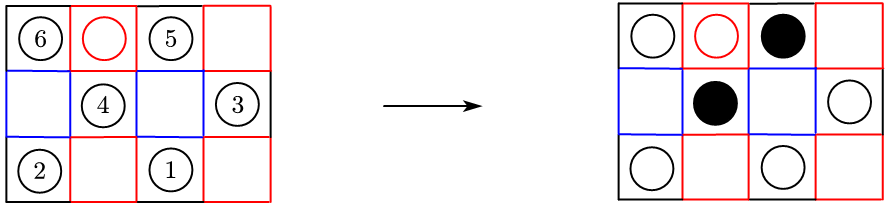}
  \end{minipage}
\caption{The Go-diagram associated with $A_{1}\widehat{\bigoplus}A_{2}$.\label{fig:GO37}}
\end{figure}
\end{example}

\subsection{Hirota $N$-solitons}
For Hirota $N$-soliton solutions, each line-soliton $[i_{k},j_{k}]$ corresponds to a vertex operator $V[A_{k}]$ with $|\mathcal{I}[A_k]|=2$, where each block $A_k\in \mathrm{Gr}(1,2)_{\ge 0}$ is of the form $(1,a_{k,j_k})$ with $a_{k,j_k}>0$. In this case, the Go-diagram associated with the $\kappa$-direct sum can be determined directly from
the chord diagram. Indeed, each pipe corresponds to a chord, and crossing chords correspond precisely to black stones in the Go-diagram. Note that for Hirota $N$-solitons,
the chord diagram is symmetric with respect to the horizontal axis.
\begin{example}\label{ex:hirota}
Consider the Hirota 4-soliton solution with the index sets
\begin{align*}
\hat{\mathcal{I}}[A_1]=\{1,4\},\qquad \hat{\mathcal{I}}[A_2]=\{2,6\},\qquad \hat{\mathcal{I}}[A_3]=\{3,8\},\qquad \hat{\mathcal{I}}[A_4]=\{5,7\}.
\end{align*}
The associated matrices $A_{k}\in\text{Gr}(1,2)_{\geq0}$ are given by
\begin{align*}
A_{1}=\left(
\begin{array}{cc}
1 & a_{1,4}
\end{array}
\right),\qquad
A_{2}=\left(
\begin{array}{cc}
1 & a_{2,6}
\end{array}
\right),\qquad
A_{3}=\left(
\begin{array}{cc}
1 & a_{3,8}
\end{array}
\right),
\qquad
A_{4}=\left(
\begin{array}{cc}
1 & a_{4,7}
\end{array}
\right),
\end{align*}
whose $\kappa$-direct sum takes the form
\begin{align*}
\widehat{\bigoplus}_{k=1}^4A_{k}=\left(
\begin{array}{cccccccc}
1 & & &a_{1,4} & & & &\\
  &1 & & & &a_{2,6} & &\\
  & & 1& & & & &a_{3,8}\\
 & & & &1 & & a_{4,7}&\\
\end{array}
\right)\in\text{Gr}(4,8).
\end{align*}
The associated chord diagram is determined by the permutation $\pi=(1,4)(2,6)(3,8)(5,7)$.
\begin{figure}[H]
  \begin{minipage}[t]{1\linewidth}
  \centering
  \includegraphics[height=1.7cm,width=5cm]{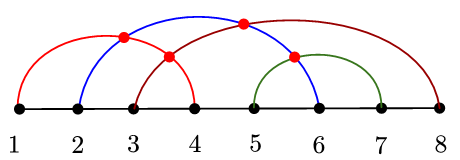}
  \end{minipage}
\end{figure}
\noindent The resulting Go-diagram is shown below.
\begin{figure}[H]
  \begin{minipage}[t]{1\linewidth}
  \centering
  \includegraphics[height=3.2cm,width=9cm]{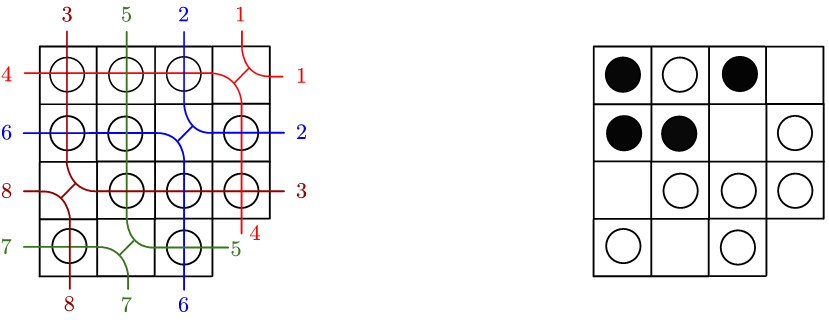}
  \end{minipage}
\end{figure}
\end{example}

\bigskip
\noindent
{\bf Acknowledgements.}
The author would like to thank Yuji Kodama for many discussions and useful suggestions on this work. The author also thanks Chuanzhong Li for his constant encouragement. This work was completed during the author's visit to Nagoya University, and the author is grateful for the opportunity and research environment provided during the visit. The visit was supported by the Yuji Kodama Scholarship at Shandong University of Science and Technology.

 \bigskip
 \noindent
{\bf Declarations.}
\begin{itemize}
\item No data availability statement applies to this article.
\item The author declares no conflict of interest.
\end{itemize}


\raggedright


\end{document}